\documentclass[11pt,fleqn]{article}

\usepackage[a4paper,text={6.125in,10.0in},centering]{geometry}
\usepackage{latexsym}
\usepackage{amsmath}
\usepackage{amssymb}
\usepackage{graphicx}
\usepackage{mathrsfs}
\usepackage{bm}
\usepackage{amsthm}
\usepackage[usenames]{color}
\usepackage{subfigure}
\usepackage{comment}

\newtheorem{theorem}{Theorem}[section]
\newtheorem{lemma}[theorem]{Lemma}

\newtheorem{cor}[theorem]{Corollary}

\newenvironment{arrayl}%
        {
        \begin{array}{@{}r@{\hskip\arraycolsep}c@{\hskip\arraycolsep}l}}%
        {\end{array}}

\newcommand{\sign}{\mathop{{\rm sign}}}
\newcommand{\uno}{u_l}

\newcommand{\utw}{u_r}
\newcommand{\pno}{p_l}

\newcommand{\ptw}{{p_r}}
\newcommand{\Pno}{\mathop{\mathscr{B}_l}}

\newcommand{\Ptw}{\mathop{\mathscr{B}_r}}

\newcommand{\bU}{\mathbf{U}}
\newcommand{\bW}{\mathbf{W}}
\newcommand{\uf}{u_{\rm f}}

\title{\textbf{A stability criterion for the non-linear wave equation with spatial inhomogeneity
}}
\author{
Christopher J.K.\ Knight\thanks{
Dept.\ of Mathematics,
University of Surrey, Guildford, Surrey, GU2 7XH
({\tt c.j.k\_knight@hotmail.co.uk})},
\and
Gianne Derks\thanks{
Dept.\ of Mathematics,
University of Surrey, Guildford, Surrey, GU2 7XH
({\tt g.derks@surrey.ac.uk})}}

\begin{document}

\maketitle

\begin{abstract}
\noindent  
In this paper the non-linear wave equation with a spatial
inhomogeneity is considered.  The inhomogeneity splits the unbounded
spatial domain into three or more intervals, on each of which the
non-linear wave equation is homogeneous.  In such setting, there often
exist multiple stationary fronts.  In this paper we present a
necessary and sufficient stability criterion in terms of the length of
the middle interval(s) and the energy associated with the front in
these interval(s).  To prove this criterion, it is shown that critical
points of the length function and zeros of the linearisation have the
same order. Furthermore, the Evans function is used to identify the
stable branch.  The criterion is illustrated with an example which
shows the existence of bi-stability: two stable fronts, one of which
is non-monotonic. The Evans function also give a sufficient
instability criterion in terms of the derivative of the length function.
\end{abstract}

\section{Introduction}
The non-linear wave equation (sometimes called the non-linear Klein-Gordon equation)
\[
u_{tt}=u_{xx}+V'(u),\quad t>0, \quad x\in\mathbb{R},
\] 
models various systems.  For instance, taking $V(u)=D(1-\cos u)$,
gives the sine-Gordon equation,
\[
u_{tt}=u_{xx}-D\,\sin{u},
\] 
which describes various physical and biological systems, including
molecular systems, dislocation of crystals and DNA
processes~\cite{BS78,D85,GJM79,W74,Y98}.  As an illustrative
example: the sine-Gordon equation is a fundamental model for long
Josephson junctions, two superconductors sandwiching a thin
insulator~\cite{PB97,S69}. In the case of Josephson junctions, the
coefficient $D$ represents the Josephson tunnelling critical current. In an
ideal uniform Josephson junction, this is a constant. But if there are
magnetic variations, e.g. because of non-uniform thickness of the
width of the insulator or if the insulator is comprised of materials
with different magnetic properties next to each other, then the
Josephson tunnelling critical current~$D$ will vary with the spatial
variable~$x$, leading to an inhomogeneous potential
$V(u,x)=D(x)(1-\cos u)$.  If there is a defect in the form of a
scratch or local thickening in the insulator then we model the
inhomogeneity by a step function $D(x)$, with $D=1$ outside the defect
and $D\neq 1$ inside, 
see~\cite{KDDS11} and references therein.

To make it easier to generalise this example to more general wave
equations, we write the model for a Josephson junction with one
inhomogeneity as an equation on the disjoint open intervals $I_l$,
$I_m$ and $I_r$ ($\mathbb{R}=\overline{\cup I_i}$):
\begin{equation}\label{WaveEq}
  u_{tt}=u_{xx}+\frac{\partial}{\partial u}V(u,x;I_l,I_m,I_r)-\alpha u_t.  
\end{equation}
Here $\alpha\geq0$ is a constant damping coefficient and the potential
$V(u,x;I_l,I_m,I_r)$ consists of three smooth ($C^3$) functions
$V_i(u)$, defined on three disjoint open intervals~$I_i$ of the real
spatial axis, such that $\mathbb{R}=\overline{\cup I_i}$.
Without loss of generality, we can write $I_l=(-\infty, -L)$,
$I_m=(-L,L)$ and $I_r=(L,\infty)$ for some $L$, by translating the $x$
variable so that the origin occurs at the centre of the middle
interval.  So, the length of the inhomogeneity or `defect' is $2L$.
In this paper we study this equation in its generality, as well as the
motivating example of Josephson junctions.  

The existence and stability of stationary fronts or solitary waves
(from now on we shall simply refer to both of these as stationary
fronts) of \eqref{WaveEq} is studied in~\cite{DDKS09,KDDS11}.  In
particular, spectral stability of the stationary fronts (which in this
case implies non-linear stability, see~\cite{KDDS11} and references
therein) is studied and a necessary and sufficient criterion is
developed for the spectral operator to have an eigenvalue zero.  The
spectral operator can be related to a self-adjoint linearisation
operator such that $\lambda$ is an eigenvalue of the spectral operator
if and only if $\Lambda=\lambda(\lambda+\alpha)$ is an eigenvalue of
the self-adjoint operator. The self-adjoint operator has real
eigenvalues and a continuous spectrum on the negative half line
bounded away from zero.  So as the largest eigenvalue passes through
zero a change of stability occurs.  The self-adjoint operator is a
Sturm-Liouville operator, which implies that the discrete eigenvalues
are simple and bounded above.  It also means that the eigenfunction
associated with the largest discrete eigenvalue will have no zeroes,
providing a tool to identify the largest eigenvalue.  In~\cite{KDDS11}
a necessary and sufficient condition for the existence of an
eigenvalue zero is derived, but this is only a necessary condition for
a change of stability.  The proof 
in~\cite{KDDS11} contains the construction of an eigenfunction and
hence the absence of zeroes can be checked to verify that the
eigenvalue zero is the largest eigenvalue. But to guarantee a change
of stability, one also needs verification that the largest eigenvalue
crosses through zero.  In this paper we will consider this crossing
question and derive a necessary and sufficient condition for the
change of stability of a stationary front.

\medskip 
As this paper builds on~\cite{KDDS11}, we start with a brief
introduction to the relevant notation and results from this paper.
Stationary fronts of~\eqref{WaveEq} are solutions of the Hamiltonian ODE
\[
0=u_{xx}+\frac{\partial}{\partial u}V(u,x;I_l,I_m,I_r),
\] 
which satisfy $u_x(x)\to0$ exponentially fast for $x\to\pm\infty$.  We
introduce the notation $p=u_x$, then the Hamiltonian $H(u,p) =
\frac{1}{2}p^2+V(u,x;I_l,I_m,I_r)$ is constant on each of the
intervals $I_i$, $i=l,m,r$. In the middle interval, we define the
Hamiltonian parameter
$g=\frac{1}{2}p^2+V_{m}(u)$ for $x\in I_m$. Writing $\uno$ and $\utw$
for the values of $u$ at the boundaries between the different
intervals ($x=-L$ and $x=L$) and $\pno$ and $\ptw$ for the values
of $p$ at the same values of $x$, the matching conditions at the
boundaries give that $u_i$ and $p_i$ can be parametrised by $g$:
\begin{equation}\label{eqrel2}
  \begin{array}{rcccl}
\frac{1}{2}\pno^2 &= &g-V_{m}(\uno) &=&V_--V_l(\uno);\\
\frac{1}{2}\ptw^2 &=& g-V_{m}(\utw)&=& V_+-V_r(\utw).    
  \end{array}
\end{equation} 
Here $V_-$ and $V_+$ are the asymptotic values of $V_l(u(x))$ and
$V_r(u(x))$ respectively, that is $V_-:=\lim_{x\to -\infty}V_l(u(x))$
and $V_+:=\lim_{x\to\infty}V_r(u(x))$.  These limits are well defined
as $u(x)$ is a front. The relations in~\eqref{eqrel2} might lead to
multi-valued functions $\uno$, etc. However, usually they are locally
well-defined, except potentially at some isolated bifurcation
points. To find those points, we define the following bifurcation functions
\begin{eqnarray*}\label{bifeq}
\Pno(g)&=&\pno(g)[V_{m}'(\uno(g))-V_l'(\uno(g))], \\
\Ptw(g)&=&\ptw(g)[V_r'(\utw(g))-V_{m}'(\utw(g))].
\end{eqnarray*}
A bifurcation occurs if $g$ satisfies $\Pno(g)=0$ or $\Ptw(g)=0$.  We
define the bifurcation values $g_{\rm bif}$ to be any value of $g$ (if
it exists) where $\Pno(g)=0$ or $\Ptw(g)=0$. These points are usually
associated with the edge of the existence interval.

Finally, the fundamental theorem of calculus enables the length of the middle
interval to be parametrised by the Hamiltonian parameter $g$.  For
instance, if $u_x(x;g)$ has no zeroes in the middle interval then
\[
2L(g)=\int_{\uno(g)}^{\utw(g)}\frac{du}{p(u,g)},
\] 
where $\uno(g)$ and $\utw(g)$ are the values of $u$ where the front
$u(x;g)$ enters resp.~leaves the middle interval.  If $u_x(x;g)$ has
zeroes in the middle interval than the expression for $L(g)$ is more
complicated but it can still be expressed explicitly, see~\cite{KDDS11}.

Now we can formulate the necessary and sufficient condition for the
existence of an eigenvalue zero from~\cite{KDDS11}.
\begin{theorem}[{\cite[Theorem 4.5]{KDDS11}}]\label{OldTheo}
 Let the front $\uf(x;g)$ be a solution of \eqref{WaveEq},
  such that all zeroes of $\partial_x\uf(x;g)$ are simple and the length of the
  middle interval of $\uf(x;g)$ is part of a smooth length curve $L(g)$.
  The linearisation operator $\mathcal{L}(g):=D_{xx}+\frac{\partial^2}{\partial
  u^2}V(\uf(x;g),x;g),~x\in \cup I_i$, associated with $\uf(x;g)$
  has an eigenvalue zero (with eigenfunction in $H^2(\mathbb{R})$) if and only if
\begin{equation}\label{conone}
\Pno(g)\Ptw(g)L'(g)=0.
\end{equation}
If $g\neq g_{\rm bif}$, then there is an eigenvalue zero (with
eigenfunction in $H^2(\mathbb{R})$) if and only if $L'(g)=0$.
\end{theorem}

As the linear operator $\mathcal{L}(g)$ is a Sturm-Liouville operator,
the eigenvalues are simple and bounded above, meaning that away from
the bifurcation points, the curves of eigenvalues,
$\Lambda(g)$ are well-defined and continuous.  
Thus if there exists an eigenvalue zero for $g=g_0\neq g_{\rm bif}$,
and the associated eigenfunction $\Psi(x)$ has no zeroes (this
eigenfunction is constructed in the proof of Theorem~\ref{OldTheo},
see \cite{KDDS11}), then a change of stability (both spectral and
non-linear) \emph{may} occur as $g$ passes from one side of $g_0$ to
the other.
Whether or not the change of stability will \emph{actually} happen,
depends on the degree of the zero of the largest eigenvalue
$\Lambda(g)$ at $g=g_0$.
Specifically, the largest eigenvalue may be strictly negative, touch
zero for some value of $g$ and then become strictly negative again,
i.e. the eigenvalue $\Lambda(g)$ may have a non-simple zero and the
front is stable for all $g$, see the left panel of
Figure~\ref{PartIITurningEValue}.
\begin{figure*}[htb]
\centering\hspace*{-0.3in}
\includegraphics[width=.6\textwidth]{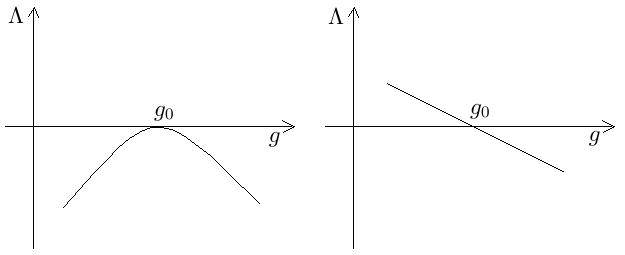}
\caption{Two possible local behaviours for the eigenvalue $\Lambda(g)$
  when $\Lambda(g_0)=0$.  In the left panel the eigenvalue
  $\Lambda(g_0)$ has a second order zero.}\label{PartIITurningEValue}
\end{figure*} 
%
In the examples in~\cite{KDDS11}, no eigenvalues of this type were
encountered; in all cases, the presence of an eigenvalue zero 
also led to a change of the sign in this eigenvalue.
Thus the eigenvalue, as a function of $g$, locally looked like the
right panel of Figure~\ref{PartIITurningEValue}.  This suggested that
this would always be the case for the inhomogeneous non-linear wave
equation~\eqref{WaveEq}.  However this suggestion is false.

In the next section we shall present an example where the length
function $L(g)$ has an inflection point. 
From~Theorem~\ref{OldTheo} it follows that an inflection point of the
length function $L(g)$ must correspond to an eigenvalue zero.  We will
show that in this example the inflection point of $L(g)$ corresponds
to a non-simple root of the eigenvalue $\Lambda(g)$ and that a change
of stability does not occur as $g$ moves from one side of $g_0$ to the
other, see Figure~\ref{infleclength}.  This result is generalised in
section~\ref{Sec.SC}, leading to a necessary and sufficient condition
for a change of stability.  In section~\ref{sec.evans}, the Evans
function is used to show that if $\Pno(g)\Ptw(g)L'(g)$ times the
product of the signs of the derivative of the front near its end
points is positive, then the front is unstable.  Combined with the
results of section~\ref{Sec.SC}, this implies that the stable branch
is the branch for which the product is negative.
To complete our exposition, in section~\ref{sec.multiple}, we
consider the wave equation with multiple middle intervals. It will be
shown that the criterion for a wave equation with one middle interval
can be used to derive a sufficient and necessary condition for a
change of stability in such systems.
\begin{figure*}[h]
\centering
\includegraphics[height=1.7in]{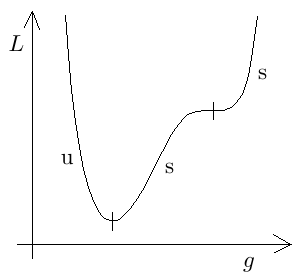}
\caption{A length curve with an inflection point and an eigenvalue with a non-simple root (see left panel of Figure~\ref{PartIITurningEValue}).  
The symbol `s' denotes a stable branch 
whilst the symbol `u' denotes an unstable branch.}\label{infleclength}
\end{figure*}

\section{Motivating Example}\label{sec.ex}
In this section we present an example in which the potential in the
inhomogeneous wave equation~\eqref{WaveEq} is such that the length
curve $L(g)$ has an inflection point at which no change of stability
occurs.  This example involves a one parameter family of potentials,
leading to a two parameter curve of length curves (the Hamiltonian
parameter~$g$ and the parameter for the potentials).  We will use the
full family to make inferences about the stability of various branches
of $L(g)$ and hence to show that the inflection point in the length
curve
does not correspond to a change of stability.  Finally we illustrate
the inferences that we made by confirming the prediction that there is
a region in the parameter space related to the potential, for which
there is
bi-stable behaviour.

The example we consider has the sine-Gordon potential (with induced
current and dissipation) in the left and right intervals ($|x|>L$) and
has a potential of the form $V_m(u)=\frac{k}{2}(u-2\pi-c)^2$ in the
middle interval ($|x|<L$), where $k$ and $c$ are parameters.  That is,
we consider stationary fronts of the following equation
\begin{equation}\label{newlabel}
 u_{tt}=u_{xx}+\frac{\partial V(u,x;L)}{\partial u}-\alpha u_t \mbox{~~where~~}
V(u,x;L):=\left\{\begin{array}{cc}
\cos(u)+\gamma u, & |x|>L;\\                                                                               
\frac{k}{2}(u-2\pi-c)^2, & |x|<L. \\
\end{array}\right.
\end{equation}
The stationary fronts join the steady states $\arcsin(\gamma)$ as
$x\to-\infty$ and $2\pi+\arcsin(\gamma)$ as $x\to\infty$.  In
Figure~\ref{clength}, typical length curves $L(g)$ are plotted for
various values of $c$ between $0$ and $-2.8$, while $\gamma$ and $k$
are kept fixed at $\gamma=0.1$ and $k=1$.  Recall that $g$ represents
the Hamiltonian in the middle interval, i.e.,
$g=\frac{1}{2}(u_x^2+k(u-2\pi-c)^2)$. The value of $\alpha$
is not relevant for the existence of stationary fronts nor for the
length curve. It only comes into play once the (linear) stability of the
stationary fronts is considered. However, for all values of
$\alpha\geq0$, the front is either always stable or always
unstable. For $\alpha=0$, the stable fronts are neutrally stable
(purely imaginary eigenvalues in the linearisation). While for
$\alpha>0$, the same fronts are now attractors.
\begin{figure*}[thb]\hspace*{-0.3in}
\centering
\includegraphics[width=.49\textwidth]{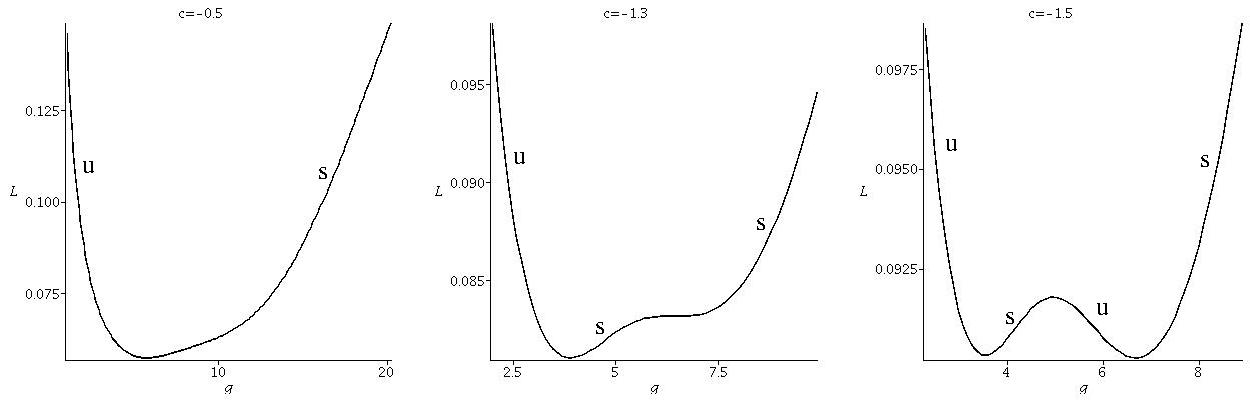}%
\includegraphics[width=.49\textwidth]{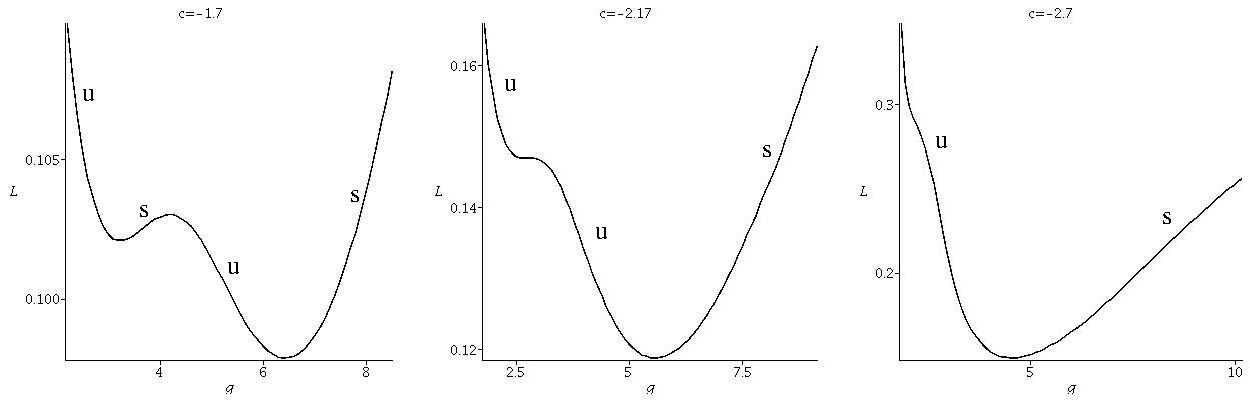}
\caption{Typical changes in the length curve when $k=1$ and
  $\gamma=0.1$ and $c$ is decreased (left to
  right) from $c=-0.5$ to $c=-2.7$.  The symbol `s' denotes a
  stable branch whilst the symbol `u' denotes an unstable
  branch.}\label{clength}
\end{figure*}
Figure~\ref{clength} illustrates that there are qualitative changes in
the length curves when $c$ is varied.  This qualitative change effects
the number of turning points of $L(g)$, changing from one to three and
back to one again as $c$ is varied.
%
%
We show this sketch for $k=1$ as it provides the full picture of what
happens as $c$ is varied.  For smaller $k$ some of these panels (right
most ones) are not seen due to a restricted $g$ existence interval.

We have included the stability of the various branches of $L(g)$ in
Figure~\ref{clength}.  The stability of the two branches for smallest
and largest values of $c$ are found by numerical simulations with
$\alpha=0.1$ for one solution on each of the branches.
Theorem~\ref{OldTheo} gives that an eigenvalue zero can only occur at
turning or inflection points of $L(g)$, hence on one branch the
stability can not change.  The stability profiles for the other panels
are then inferred from these two panels, using the fact that $V$ is
continuous in~$c$, hence $L$ will be continuous in~$c$.  For example,
to get the stability of the branches for the second image from the
left, note that the inflection point occurs away from the turning point
so the stability of the branches in the locale of the turning point
cannot change (no eigenvalues have crossed zero for $g$ in the
neighbourhood of the turning point).  Furthermore, the stability at
the outer parts of the branches (largest and smallest values of $g$)
can not change as $L'(g)\neq 0$ at the outer parts. Therefore the
only possibility for the stability of the branches is as shown in this
second panel and hence the stability doesn't change at the inflection
point (as $g$ is varied and $c$ is kept fixed).
A similar argument can be used for the second image on the right. Now
the stability in the middle images follows from matching the outer
panels on the left and right. Going from the second to the third image
from the left, continuity does not provide any conclusion about the
stability of the third branch in first instance. Similarly we don't
know the stability of the second branch in the third picture on the
right. However, by comparing these two images and using continuity, we
can conclude the stability of these branches.

To illustrate the inferences about the stability of the various
branches, we observe that Figure~\ref{clength} suggests that there are
values of $c$ with two stable stationary fronts. It is easier to
visualise this bi-stable behaviour for smaller values of $k$, so we
focus on $k=0.6$.  Recall that this change in $k$ has the effect of
removing some of the panels (from the right) in Figure~\ref{clength}.
For $k=0.6$ and $\gamma=0.1$, there is an inflection point on the
right branch when $c=\widehat{c}\approx0.75$.  If $c=-2$ (and $k=0.6$,
$\gamma=0.1$) then the length curve has three turning points, see
Figure~\ref{fig.3},
\begin{figure*}[h]
\centering\hspace*{-0.3in}
\includegraphics[height=2in]{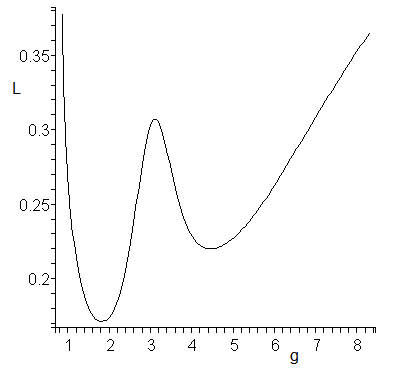}
\caption{The curve $L(g)$ for $\gamma=0.1$, $k=0.6$ and $c=-2$ showing
  that there are four stationary fronts for $L=0.26$.}\label{fig.3}
\end{figure*} 
and for $L=0.26$ there are four stationary fronts, two of which are
stable.
In the left panel of Figure~\ref{fig.4} we show the four stationary
fronts for $L=0.26$.
The red non-monotonic front corresponds to the right most branch in
Figure~\ref{fig.3}.  The fact that two of the fronts are stable is
confirmed by simulating equation~\eqref{newlabel}, starting with different
initial conditions, see the middle and right panels of
Figure~\ref{fig.4}.  The initial conditions used are the (unique)
stable front for a slightly larger length value ($L= 0.33$), which
converges to the stable non-monotonic stationary front (right
plot). And the other initial condition is the stationary front
solution to the sine-Gordon equation, which converges to the other
stable front (middle plot). It is interesting to see that the
non-monotonic stationary front is indeed stable.
\begin{figure*}[h]
\centering
\includegraphics[height=1.7in]{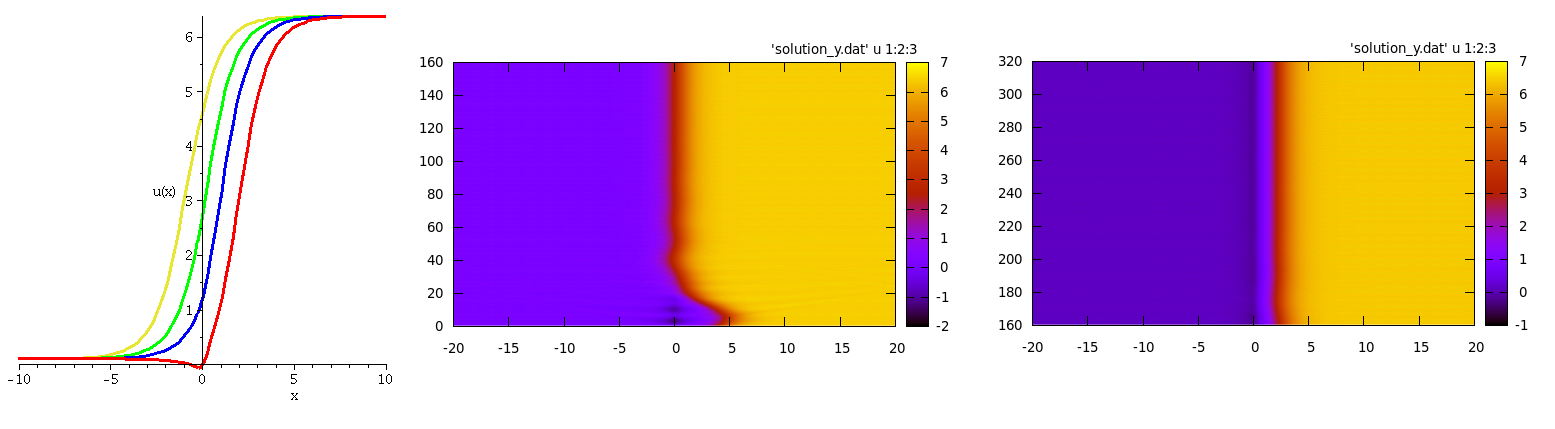}
\caption{Left: the four stationary fronts for $\alpha=0.1$, $k=0.6$,
  $\gamma=0.1$, $c=-2$ and $L=0.26$. Middle and Right: Simulations
  of~\eqref{newlabel} converging to the two stable stationary fronts
  (second (yellow) and fourth (red) front in the left
  plot).}\label{fig.4}
\end{figure*} 

This simulation also confirms the inference that no change of
stability occurs at the inflection point.  
In the next section we will show that this is generically true.


\section{A Necessary and Sufficient Criterion for Stability Change}\label{Sec.SC}

In the example of the previous section we saw an inflection point of
the length function $L(g)$ which corresponded to a non-simple zero of
the eigenvalue function $\Lambda(g)$.  That is: at the inflection point
$g=\widehat{g}$, it holds that
$L'(\widehat{g})=0=\Lambda(\widehat{g})$ and also
$L''(\widehat{g})=0=\Lambda'(\widehat{g})$.  Here we see a hint of the
relationship between the order of a zero of $L'(g)$ and the order of
the corresponding zero of $\Lambda(g)$.  We make this explicit in the
following lemma.
\begin{lemma}\label{HysLem}
  Away from the bifurcation point $g_{\rm bif}$, the order of a zero
  of $\Lambda(g)$ is the same as the order of the associated
  stationary point of $L(g)$.  For instance if $L(g)$ has a stationary
  point at $g=g_0\neq g_{\rm bif}$, with $L'(g_0)=0=L''(g_0)$
  and $L'''(g_0)\neq0$, then $g_0$ is a second order root of
  $\Lambda(g)$ (i.e., $\Lambda(g_0)=0=\Lambda'(g_0)$ and
  $\Lambda''(g_0)\neq0$) and visa verse.
\end{lemma}
\begin{proof}\mbox{}
  \newline Let $L(g)$ have a stationary point at $g=g_0$, then
  Theorem~\ref{OldTheo} implies that the linearisation operator
  $\mathcal{L}$ associated with the front $u(x;g_0)$ has an eigenvalue
  $\Lambda(g_0)=0$.

  As a thought experiment (motivated by the previous section),
  consider smooth perturbations of the potential $V_m(u)$ which have
  the effect of perturbing the associated $L(g)$ curve.  The
  perturbation may lead to a change in the number of
  stationary points of $L(g)$ (in the locality of $g_0$).
  As it is perturbations of the potentials which cause this change in
  $L(g)$, Theorem~\ref{OldTheo} still holds, meaning that the
  stationary points of $L(g)$ are still associated with zeroes of
  $\Lambda(g)$.  Thus the original curve $\Lambda(g)$ will be smoothly
  perturbed in such a way that it has the same number of zeroes as the
  perturbed curve $L(g)$ has stationary points.
  Away from any bifurcation points ($g_{\rm bif}$), a smooth
  perturbation of $V_m(u)$ leads to a smooth perturbation of $L(g)$.
  If $L(g)$ has a turning point of first order  at $g=g_0$, then any
  small smooth perturbation will lead to $L(g)$ still having exactly one
  turning point (in the locality of $g_0$).  Similarly, if $\Lambda(g)$ has a
  simple zero (right panel in Figure~\ref{PartIITurningEValue}) then
  any smooth perturbation to $\Lambda(g)$ will still (locally)
  have one zero.  Whilst if $\Lambda(g)$ has a second order zero (left
  panel in Figure~\ref{PartIITurningEValue}) a small smooth perturbation
  would generically result in either no zeroes of $\Lambda(g)$ or in
  two zeroes of $\Lambda(g)$, which is not consistent with a first
  order turning point of $L(g)$.  Thus we can conclude that a first
  order turning point of $L(g)$ cannot be not associated with a second order
  zero of $\Lambda(g)$, and similarly, it cannot be associated with an
  even order zero of $\Lambda(g)$.

  To make those ideas formal, we embed the middle potential $V_m(u)$
  smoothly in a larger family $V_m(u,\epsilon)$, where $\epsilon$ is a
  small parameter and $V_m(u,0)=V_m(u)$. 
  For $\epsilon$ small and $g$ near $g_0$, this embeds the family of
  fronts $\uf(x;g)$ in the family $\uf(x;g,\epsilon)$ with associated
  length curves $L(g,\epsilon)$ and the eigenvalue
  $\Lambda(g,\epsilon)$.  As shown in~\cite{KDDS11}, the length $L(g)$
  is the sum of integrals of the form
  $\displaystyle\int_I\textstyle\frac{du}{2\sqrt{g-V_m(u)}}$. Hence
  the embedding $V_m(u,\epsilon)$ can be chosen such that
\begin{equation}\label{condV}
\frac{\partial^2 L}{\partial \epsilon\partial
    g}(g_0,0)\neq0 \mbox{ and } 
\left\langle\left. 
\frac{D}{D\epsilon}\left(\frac{\partial^2}{\partial u^2}
  V(\uf(g_0,\epsilon),\epsilon)\right)\right|_{\epsilon=0}
\psi(g_0,0)\,,\,\psi(g_0,0) \right\rangle_{L_2} \neq
0, 
\end{equation}
implying that 
\[
\frac{\partial \Lambda}{\partial  \epsilon}(g_0,0)\neq0.
\] 
To see this inequality, differentiating
the eigenvalue equation $[\mathcal{L}(g_0,\epsilon) -
\Lambda(g,\epsilon)]\psi(g_0,\epsilon)=0$ with respect to $\epsilon$,
evaluating at $\epsilon=0$, and taking the $L_2$ inner product with
the eigenfunction $\psi(g_0,0)$ shows that
\[
\left\langle\left. \left[\frac{D}{D\epsilon}\left(\frac{\partial^2}{\partial
    u^2} V(\uf(g_0,\epsilon),\epsilon)\right)\right|_{\epsilon=0}- 
\frac{\partial \Lambda}{\partial \epsilon}(g_0,0) \right]\,
\psi(g_0,0) \,,\, \psi(g_0,0)\right\rangle_{L_2}=0. 
\]
Hence the second inequality in~\eqref{condV} implies that
$\frac{\partial \Lambda}{\partial \epsilon}(g_0,0)\neq0$.

From Theorem~\ref{OldTheo} we know that $\frac{\partial L}{\partial
  g}(g,\epsilon)=0 \Leftrightarrow \Lambda(g,\epsilon)=0$.  To see
that this implies that the order of the critical point of $L(g)$ and
the order of the zero of $\Lambda(g)$ is the same at $g=g_0$, we start
with showing the existence of a curve of critical points $g_L(\epsilon)$
with $\frac{\partial L}{\partial g}(g_{L}(\epsilon),\epsilon)=0$ and
$g_L(0)=g_0$.

First we consider the non-degenerate case: if $\frac{\partial
  L}{\partial g}(g_0,0)=0$ and $\frac{\partial^2 L}{\partial
  g^2}(g_0,0)\neq 0$ then the Taylor series for $\frac{\partial
  L}{\partial g}(g,\epsilon)$ about $(g_0,0)$ is
\[
\frac{\partial L}{\partial g}(g,\epsilon)=(g-g_0)\frac{\partial^2 L}{\partial g^2}(g_0,0)
+\epsilon \frac{\partial^2 L}{\partial g \partial \epsilon}(g_0,0)
+O\left(|\epsilon+(g-g_0)|^2\right)
\]
Since $\frac{\partial^2 L}{\partial g^2}(g_0,0)\neq0$, the implicit
function theorem gives that there exist a unique curve $g_L(\epsilon)$
for $\epsilon$ near zero such that $\frac{\partial L}{\partial
  g}(g_{L}(\epsilon),\epsilon)=0$ and $g_L(\epsilon)=g_0-\epsilon\,
\left(\frac{\partial^2 L}{\partial g^2}(g_0,0)\right)^{-1}
\frac{\partial^2 L}{\partial g \partial \epsilon}(g_0,0)+
O(\epsilon^2)$. Differentiating $\frac{\partial L}{\partial
  g}(g_{L}(\epsilon),\epsilon)=0$ with respect to $\epsilon$ and
evaluating at $\epsilon=0$ gives
\[
g_L'(0)= -  \left(\frac{\partial^2 L}{\partial
    g^2}(g_0,0)\right)^{-1} \frac{\partial^2 L}{\partial g \partial
  \epsilon}(g_0,0)  \neq 0.
\]
Since $\frac{\partial L}{\partial g}(g_{L}(\epsilon),\epsilon)=0$,
Theorem~\ref{OldTheo} gives that $\Lambda(g_L(\epsilon),\epsilon)=0$
too. Differentiating this expression with respect to $\epsilon$ and
evaluating at $\epsilon=0$ shows
\[
0 = \frac{\partial\Lambda}{\partial g}(g_0,0)g_L'(0) +
\frac{\partial \Lambda}{\partial \epsilon}(g_0,0), \mbox{
  hence }
\frac{\partial\Lambda}{\partial g}(g_0,0) = 
\frac{\partial \Lambda}{\partial \epsilon}(g_0,0)  
\left(\frac{\partial^2 L}{\partial g \partial \epsilon}(g_0,0)
\right)^{-1}
\frac{\partial^2 L}{\partial g^2}(g_0,0)\neq 0.
\]
And we can conclude that $g_0$ is a first order zero of the eigenvalue
$\Lambda(g)$, hence a first order turning point of $L(g)$ is associated
with a simple zero of $\Lambda(g)$.

\smallskip 
Next we consider the higher order critical points of $L(g)$.  Let $n$
be the smallest value ($n\geq 2$) such that $\frac{\partial^{n+1}
  L}{\partial g^{n+1}}(g_0,0)\neq0$, that is $\frac{\partial^k
  L}{\partial g^k}(g_0,0)=0$ for $k=1,\ldots, n$.  Then a Taylor series
gives
\[
\frac{\partial L}{\partial
  g}(g,\epsilon)=\frac{1}{n!}(g-g_0)^n\frac{\partial^{n+1} L}{\partial
  g^{n+1}}(g_0,0) +\epsilon \frac{\partial^2 L}{\partial g \partial
  \epsilon}(g_0,0)
+O\left(\epsilon^2+\epsilon(g-g_0)+(g-g_0)^{n+1}\right).
\] 
In order to be able to apply the implicit function theorem and get a
smooth curve of critical points $g_L$, we replace $\epsilon$ with the
new variable $\eta$ such that
\[
\eta^n = 
{-\epsilon \left(\frac{\partial^{n+1} L}{\partial
  g^{n+1}}(g_0,0) \right)^{-1}\frac{\partial^2 L}{\partial g \partial
  \epsilon}(g_0,0)}.
\]
If $n$ is even, then we restrict $\epsilon$ to only positive or only
negative values. However, $\eta$ is allowed to have both positive and
negative values. We define 
\[
\widetilde{L}(g,\eta)=
L\left(g,-\eta^n\left(\frac{\partial^2 L}{\partial g \partial
    \epsilon}(g_0,0) \right)^{-1}\frac{\partial^{n+1} L}{\partial
  g^{n+1}}(g_0,0)\right). 
\]
Now the implicit function theorem gives that there exist a unique
curve $g_L(\eta)$ for $\eta$ near zero such that $\frac{\partial
  \widetilde L}{\partial g}(g_{L}(\eta),\eta)=0$ and
$g_L(\eta)=g_0+\eta+ O(\eta^2)$, hence $g_L'(0)=1$.  Now we can
proceed as before. As 
$0=\frac{\partial \widetilde L}{\partial g}(g_L(\eta),\eta) = 
\left. \frac{\partial L}{\partial g}\left(g,-\eta^n\left(
\frac{\partial^2 L}{\partial g \partial \epsilon}(g_0,0)\right)^{-1}
\frac{\partial^{n+1} L}{\partial
  g^{n+1}}(g_0,0)\right)\right|_{g=g_L(\eta)}$, Theorem~\ref{OldTheo}
gives that 
$\Lambda\left(g_L(\eta),-\eta^n\left(\frac{\partial^2 L}{\partial
      g \partial \epsilon}(g_0,0) \right)^{-1}\frac{\partial^{n+1}
    L}{\partial g^{n+1}}(g_0,0)\right)=0$. Differentiating this
expression $n$ times with respect to $\eta$ and evaluating at
$\eta=0$, gives that for $k=1,\ldots,n-1$
\[
\frac{\partial^k\Lambda}{\partial g^k}(g_0,0)=0 \mbox{ and }
\frac{\partial^n\Lambda}{\partial g^n}(g_0,0) =
n!\,\frac{\partial\Lambda}{\partial \epsilon}(g_0,0)\left(\frac{\partial^2 L}{\partial
      g \partial \epsilon}(g_0,0) \right)^{-1}\frac{\partial^{n+1}
    L}{\partial g^{n+1}}(g_0,0) \neq0.
\]
Thus we have shown that an $n$-th order critical point of $L(g)$
corresponds to a $n$-th order zero of $\Lambda(g)$.
\end{proof}
Lemma~\ref{HysLem} equates the order of zeroes of $L'(g)$ with the
order of zeroes of $\Lambda(g)$.  This, together with the results of
\cite{KDDS11} and Sturm-Liouville theory, gives a necessary and
sufficient condition for a change of stability to occur for $g$ away
from any (existence) bifurcation point $g_{\rm bif}$.
\begin{theorem}\label{New1}
  Let the front $\uf(x;g)$, $g\neq g_{\rm bif}$, be a solution of
  \eqref{WaveEq}, such that all zeroes of $\partial_x\uf(x;g)$ are simple and
  the length of the middle interval of $\uf(x;g)$ is part of a smooth
  length curve $L(g)$.  There is a change in stability of $\uf(x;g)$ at
  $g=\widehat{g}\neq g_{\rm bif}$ if and only if
\begin{enumerate}
 \item $L'(\widehat{g})=0$;
\item the order of the zero $\widehat{g}$ of $L'(g)$ is odd;
\item the eigenfunction  $\Psi(\widehat{g})$ has no zeroes.
\end{enumerate}
If just i) and ii) are satisfied then the number of positive eigenvalues of the linearisation 
operator $\mathcal{L}(g)$ changes by one as $g$ crosses $\widehat{g}$.
\end{theorem}
\begin{proof}\mbox{}\newline
  As stated in Theorem~\ref{OldTheo}, the results of \cite{KDDS11} (in
  particular Theorem 4.5) gives that the linearisation operator
  $\mathcal{L}(g)$ has an eigenvalue zero for $g=\widehat{g}$ if and only
  if $L'(\widehat{g})=0$.  From Lemma~\ref{HysLem}, we can now conclude
  that if the order of the zero $\widehat{g}$ of $L'(g)$ is odd then the
  sign of $\Lambda(g)$ changes as $g$ passes from one side of
  $\widehat{g}$ to the other.  That is, the number of positive eigenvalues
  of the linearisation operator $\mathcal{L}(g)$ changes by one.  If
  the order of the zero $\widehat{g}$ of $L'(g)$ is even then there is the
  same number of positive eigenvalues either side of $\widehat{g}$.  A
  change of stability occurs at an eigenvalue zero if it is the
  largest eigenvalue and a positive eigenvalue is lost/gained as $g$
  moves from one side of $\widehat{g}$ to the other.  By Sturm-Liouville
  theory, an eigenvalue is the largest eigenvalue if and only if its
  associated eigenfunction $\Psi(\widehat{g})$ has no zeroes.  
\end{proof}
Note that the eigenfunction $\Psi(\widehat{g})$ is constructed
explicitly in 
\cite{KDDS11}, making it easy to check criterion~\emph{iii)} and that
non-linear stability can be concluded from linear stability.

\section{The Evans Function and Branch Stability}
\label{sec.evans}

Now we have characterised the fronts at which a change of stability
will occur, a next question is how the stable and unstable branch can
be identified. Due to translational symmetry, evolutionary homogeneous
Hamiltonian wave equations like the sine-Gordon equation, the
nonlinear Schrodinger (NLS) equation or the Korteweg-de Vries (KdV)
equation have families of solitary wave solutions, which can be
characterised as constrained critical points of the Hamiltonian on
level sets of the generalised momentum, i.e, the constant of motion
associated with the translational symmetry. For such equations, the
orbital (in)stability of the solitary waves can be associated with the
slope of the momentum-velocity curve. For the NLS equation this
criterion is known as the Vakhitov-Kolokolov condition~\cite{VK73}. A
general criterion for a large class of Hamiltonian wave equations with
symmetry was derived by Grillakis, Shatah and Strauss~\cite{GSS87,
  GSS90}. The proof builds on the fact that the solitary wave is a
constrained critical point with a finite dimensional negative
subspace. Pego and Weinstein~\cite{PW92} proved a linear instability
criterion for a larger class of Hamiltonian wave equations by using
the socalled Evans function.

The Evans function was named by Alexander, Gardner and
Jones~\cite{AGJ90} after J.W. Evans, who used the concept to study the
stability of nerve impulses~\cite{Evans75}. To define the Evans
function, one considers the eigenvalue problem associated with the
linear stability of stationary or travelling wave as a boundary value
problem for an ordinary differential equation. An eigenvalue $\Lambda$
will exist if the stable subspace for $x\to+\infty$ intersects with
the unstable subspace for $x\to-\infty$. The angle between those two
subspaces is measured by the Evans function, a Wronskian-like analytic
function of~$\Lambda$.

The inhomogeneity in our wave equation breaks the translational
symmetry and the criteria derived for homogeneous equations do not
apply anymore. However, our stability criterion is still related to a
vanishing derivative, namely $L'(g)=0$. Looking at
Figure~\ref{clength}, one could expect that $L'(g)>0$ might be a
necessary condition for the stability of the stationary waves.
However, if one considers the examples in~\cite{DDKS09}, there is one
example for which the stable fluxons have $L'(g)<0$ (Corollary~6 and
Figure~9) and one example for which there is a curve of stable fluxons
and the curve folds twice. On this curve, we have both $L'(g)<0$ and
$L'(g)>0$ (Theorem~8 and Figure~17).  

So on first view, there doesn't seem to be a relation between the
slope of the $g$-$L$ curve and stability. However, a closer look at our
criterion for the existence of an eigenvalue zero in
Theorem~\ref{OldTheo} shows that it involves the product of $L'(g)$
and the bifurcation functions $\Pno(g)$ and $\Ptw(g)$.  Reflecting on
the proof the existence of the eigenvalue zero in~\cite[Theorem
4.5]{KDDS11}, we note that it gives explicit expressions for the
stable and unstable subspaces. In other words, it gives all
ingredients to find the Evans function at $\Lambda=0$.  The Evans
function for large values of $\Lambda$ can also be determined and by
combining those two facts, the following instability criterion can be
proved.
\begin{theorem}\label{th.evans}
  Let the front $\uf(x;g)$ be a solution of \eqref{WaveEq},
  such that all zeroes of $\partial_x\uf(x;g)$ are simple and the length of the
  middle interval of $\uf(x;g)$ is part of a smooth length curve $L(g)$.
Define $\Pi(g)$ to be positive if $(\uf)_x(x)$ has the same sign for
$x\to+\infty$ and $x\to -\infty$ and negative if there are different
signs, i.e., 
\[
\Pi(g) = \sign\left(
\lim_{x\to\infty}e^{\left(\sqrt{\Lambda-V_r''(u_r)}-\sqrt{\Lambda-V_l''(u_l)}\right)\,x}
\, (\uf)_x(x)\,(\uf)_x(-x)\right);
\]
If 
\[
\Pi(g)\,\Pno(g)\, \Ptw(g)\, L'(g) >0,
\]
then the front $\uf(x;g)$ is linearly unstable, i.e., the
linearisation operator $\mathcal{L}(g)$ has a strictly positive eigenvalue.
\end{theorem}
\begin{proof}
The eigenvalue problem $\mathcal{L}(g) \Psi =\Lambda \Psi$ can be
written as the following first order system of ODEs
\begin{equation}\label{eq.system}
\bU_x = A(x;\Lambda,g)\, \bU\,, \mbox{ with }
\bU =
\begin{pmatrix}
  \Psi\\\Psi_x
\end{pmatrix} \mbox{ and } 
A =
\begin{pmatrix}
  0& 1\\ \Lambda-\frac{\partial^2V(\uf(x;g),x;g)}{\partial u^2}&0
\end{pmatrix}.
\end{equation}
The matrix $A(x;\Lambda,g)$ is piecewise smooth in $x$ as there are
jumps at $x=\pm L$. An eigenfunction is a continuous  solution
$\bU(x;\Lambda,g)$. As $A$ is only piecewise smooth, we can't expect
$\bU$ to be differentiable.
 
By construction, the front $\uf$ decays exponentially fast to their
steady states for $x\to\pm\infty$. If we denote these steady states by
$u_l$ (for $x\to-\infty$) and $u_r$ (for $x\to+\infty$), then these
steady states are saddles of the ODEs $u_{xx} + V_i'(u)$, $i=l$ or
$r$. Note that these steady states do not depend on $g$, the
Hamiltonian in the middle interval. The fact that the steady states
are saddles implies that $V_i''(u_i)<0$ for $i=l,r$.  Thus for
$x\to\pm\infty$, the matrix $A$ converges to
\[
A_{\pm \infty} (\Lambda)  = \lim_{x\to\pm\infty} A(x;\Lambda,g) =
\begin{pmatrix}
  0& 1\\ \Lambda-{V_i''(u_i)}&0
\end{pmatrix}, \mbox{ with $i=l$ for $-\infty$ and $i=r$ for $+\infty$.}
\]
For $\Lambda\geq0$, the term $(\Lambda-{V_i''(u_i)})$ is strictly positive, and
hence the eigenvalues of $A_{\pm\infty}(\Lambda)$ are
$\pm\sqrt{\Lambda-{V_i''(u_i)}}$, thus one positive and one negative
eigenvalue. Associated eigenvectors are
\[
\bU_\pm^i (\Lambda)=
\begin{pmatrix}
  1\\ \pm\sqrt{\Lambda-{V_i''(u_i)}}
\end{pmatrix}.
\]
The exponential decay of the  front $\uf$ to the steady states at
infinity ($u_l$, respectively $u_r$) implies that there is some $C>0$,
independent of $\Lambda$, such that
\[
\int_{-\infty}^{+\infty} \| R(x;\Lambda,g)\|\,dx <
C, \mbox{where } 
R(x;\Lambda,g) = \left\{
\begin{arrayl}
A(x;\Lambda,g)-A_{-\infty}(\Lambda),&&x\leq L,\\ 
A(x;\Lambda,g)-{A}_{+\infty}(\Lambda), &&x>L.
\end{arrayl}
\right.
\]
Now Levinson theory \cite{Coppel65,Eastham89}, \cite[Proposition
1.2]{PW92} implies that that there exists unique
solutions~$\bU_+(x;\Lambda,g)$ and~$\bU_-(x;\Lambda,g)$ of the ODE
system~\eqref{eq.system}, which are analytic in $\Lambda$ for
$\Re(\Lambda)\geq 0$ and satisfy
\[
\lim_{x\to+\infty} e^{\sqrt{\Lambda-{V_r''(u_r)}}\,x}\,\bU_+(x;\Lambda,g)=
\bU^r_-(\Lambda)\,, \quad
\lim_{x\to-\infty} e^{-\sqrt{\Lambda-{V_l''(u_l)}}\,x}\,\bU_-(x;\Lambda,g)=
\bU^l_+(\Lambda)\,.
\]
Thus the eigenvalue problem $\mathcal{L}(g) \Psi =\Lambda \Psi$ has a
solution in $H^1(\mathbb{R})$ if and only if the solutions
$\bU_+(x;\Lambda,g)$ and $\bU_-(x;\Lambda,g)$ are identical up to
scaling. The Evans function measures the angle between
$\bU_+(x;\Lambda,g)$ and $\bU_-(x;\Lambda,g)$  and can defined as the
determinant of the matrix with those two vectors as columns:
\[
D(\Lambda,g) = \det 
\begin{pmatrix}
\bU_+(x;\Lambda,g) \mid \bU_-(x;\Lambda,g)  
\end{pmatrix}.
\]
Note that there are many equivalent definitions of the Evans function,
but, for our proof here, this definition is most convenient. The fact that
the Evans function does not depend on~$x$ follows from the
Abel-Liouville Theorem and $\mbox{Tr}(A)=0$. As said before, the
eigenvalue problem is a Sturm-Liouville problem and if there are any
eigenvalues, they have to be real. Thus we can concentrate on the
Evans function for $\Lambda\in\mathbb{R}$.

\smallskip 
Next we will determine the Evans function for $\Lambda=0$, using the
expressions for $U_+(L;0,g)$ and $U_-(L;0,g)$ as implicitly given in
the proof of Theorem~4.5 in~\cite{KDDS11}. First we observe that the
derivative of the front $(\uf)_x$ satisfies $\mathcal{L}(g) (\uf)_x
=0$ on the three intervals $|x|>L$ and $|x|<L$. However, usually this
is not an $H^1(\mathbb{R})$ function, so it is not an
eigenfunction. However, on the intervals $|x|>L$, this function is
smooth. Moreover, the vector function associated with this function is
exponentially decaying, so we get
\[
U_+(x;0,g) = C_+
\begin{pmatrix}
  (\uf)_x\\(\uf)_{xx}
\end{pmatrix}, \mbox{ with } C_+(g)
=\lim_{x\to\infty}e^{\sqrt{\Lambda-{V_r''(u_r)}}\,x}\,  (\uf)_x(x);
\mbox{ for } x>L,
\]
and
\[
U_-(x;0,g) = C_-
\begin{pmatrix}
  (\uf)_x\\(\uf)_{xx}
\end{pmatrix}, \mbox{ with } C_-(g)
=\lim_{x\to-\infty}e^{-\sqrt{\Lambda-{V_l''(u_l)}}\,x}\,  (\uf)_x(x).
\mbox{ for } x<-L,
\]
In the middle interval, the solution is a linear combination of
functions of the form 
\[
(\uf)_x \mbox {  and } \frac{1}{(\uf)_x}\int \frac{dx}{(\uf)_x^2}.
\]
 We focus on the least technical case, when
$\uf$ is monotonic on the middle interval hence $(\uf)_x(x)\neq 0$ for
$|x|<L$. Details for more technical case are similar and can be
deducted in a similar way from the details in the proof
in~\cite{KDDS11}.

If $\uf$ is monotonic, then \cite[p.\ 419-420]{KDDS11} gives for the middle
interval 
\[
U_-(x;0,g) = C_-
\begin{pmatrix}
\displaystyle  (\uf)_x(x) + \Pno\, (\uf)_x(x)\int_{-L}^x\frac{d\xi}{(\uf)_x^2(\xi)}\\[4mm]
\displaystyle (\uf)_{xx}(x) + \Pno\, (\uf)_{xx}(x)\int_{-L}^x\frac{d\xi}{(\uf)_x^2(\xi)}
+ \frac{\Pno}{(\uf)_x(x)}
\end{pmatrix}.
\]
So now we can determine the Evans function at $\Lambda=0$ by
evaluation at $x=L$
\[
D(0,g) = C_-C_+\,\det
\begin{pmatrix}
  \ptw & \ptw(1 + \Pno\,I)\\ 
-V_r'(\utw) & -V_m'(\utw) (1+\Pno\,I)+\frac{\Pno}{p_r}
\end{pmatrix} = C_-C_+\,(\Pno+\Ptw +\Pno\Ptw\,I),
\]
where $I$ stands for
$I(g)=\displaystyle\int_{-L}^L\frac{d\xi}{(\uf)_x^2(\xi;g)}$. From~\cite[eq.\
(32)]{KDDS11}, it follows that $\Pno(g)+\Ptw(g)
+\Pno(g)\Ptw(g)\,I(g)=-2L'(g)\Pno(g)\Ptw(g)$, hence
\begin{equation}\label{eq.evans0}
D(0,g)=-2C_-(g)C_+(g)\,L'(g)\Pno(g)\Ptw(g).
\end{equation}

Next we consider the limit of the Evans function $D(\Lambda,g)$ for
$\Lambda\to\infty$. First we observe that the matrix $A(x;\Lambda,g)$
can be written as 
\[
A(x;\Lambda,g) = \sqrt{\Lambda}\, M^{-1}(\Lambda) \, [B_0 -
B_1(x;\Lambda,g)] \, M(\Lambda),\]
with
\[
B_0 =\begin{pmatrix}
  0&1\\1&0
\end{pmatrix}, \quad
B_1 (x;\Lambda,g)=\frac1{\Lambda}\,\begin{pmatrix}
  0&0\\{\frac{\partial^2V(\uf(x;g),x;g)}{\partial u^2}}&0
\end{pmatrix}, \quad
M (\Lambda) =\begin{pmatrix}
  \sqrt\Lambda&0\\0&1
\end{pmatrix}.
\]
Introducing the new spatial coordinate $y=\sqrt\Lambda x$ and vector
function $\bW(y) = \bU\left(\frac{y}{\sqrt\Lambda}\right)$, we see
that the ODE~\eqref{eq.system} can be rewritten as
\begin{equation}\label{eq.sys_as}
\bW_y = \left[B_0 - B_1
  \left(\frac{y}{\sqrt\Lambda};\Lambda,g\right)\right]\,\bW.
\end{equation}
As the front $\uf(x,g)$ is bounded in $x$, it follows immediately that
${\frac{\partial^2V(\uf(x;g),x;g)}{\partial u^2}}$ is bounded in~$x$
and $B_1(x;\Lambda,g)=\mathcal{O}(\Lambda^{-1})$, uniform in~$x$ and
$\Lambda$.  With the Roughness Theorem (Coppel~\cite{Coppel78}), it
follows that  the stable and unstable subspaces of~\eqref{eq.sys_as} are
order~$\Lambda^{-1}$ close to those of the constant system $\bW_y =
B_0\bW$.  Thus for $\Lambda$ large, the Evans function is
\[
D(\Lambda,g) = \det
\begin{pmatrix}
  \bW_+(\Lambda\,x)\mid \bW_-(\Lambda\,x)
\end{pmatrix} = \det
\begin{pmatrix}
1&1 \\-1&1
\end{pmatrix}+\mathcal{O}(\Lambda^{-1}) = 
2+\mathcal{O}(\Lambda^{-1}) .
\]
For $\Lambda\geq 0$, the Evans function is analytic and therefore also
continuous. So the intermediate value theorem gives that if the
Evans function at $\Lambda=0$ is negative, there must be some
$\Lambda>0$ at which the Evans function vanishes. Thus this value
of $\Lambda$ is a  positive eigenvalue of the linearisation operator
and the front $\uf(x;g)$ is unstable.
\end{proof}

Combining Theorem~\ref{th.evans} with Theorem~\ref{New1} from the
previous section, gives the following criterion which characterises a
stable branch of fronts.
\begin{cor}\label{cor.branch}
    Let the front $\uf(x;g)$, $g\neq g_{\rm bif}$, be a solution of
  \eqref{WaveEq}, such that all zeroes of $\partial_x\uf(x;g)$ are simple and
  the length of the middle interval of $\uf(x;g)$ is part of a smooth
  length curve $L(g)$.  There is a change in stability of $\uf(x;g)$ at
  $g=\widehat{g}\neq g_{\rm bif}$ if and only if
\begin{enumerate}
 \item $L'(\widehat{g})=0$;
\item the order of the zero $\widehat{g}$ of $L'(g)$ is odd;
\item the eigenfunction  $\Psi(\widehat{g})$ has no zeroes.
\end{enumerate}
Furthermore, the stable fronts are on the branch with $\Pi(g)\,\Pno(g)\,
\Ptw(g)\, L'(g) <0$.   
\end{cor}

\section{Multiple Middle Intervals}\label{sec.multiple}

So far we have considered an inhomogeneous wave equation with three
intervals. A natural extension is an inhomogeneous wave equation with
$N+2$ intervals, i.e.,
\begin{equation}\label{WaveEqN}
  u_{tt}=u_{xx}+\frac{\partial}{\partial u}V(u,x;I_l,I_{m_1}\ldots,I_{m_N},I_r)-\alpha u_t.  
\end{equation}
The potential $V(u,x;I_l,I_{m_1}\ldots,I_{m_N},I_r)$ consists of $N+2$
smooth ($C^3$) functions $V_i(u)$, defined on $N+2$ disjoint open
intervals~$I_i$ of the real spatial axis, such that
$\mathbb{R}=\overline{\cup I_i}$.
The $N$ middle intervals have lengths $L_1,\ldots, L_N$ and associated
Hamiltonians $g_1,\ldots, g_N$. In~\cite[Theorem 6.1]{KDDS11}, it is
shown that, away from the existence bifurcation points, the
linearisation about a front $\uf(g_1,\ldots,g_N)$ has an eigenvalue
zero if and only if the determinant of the Jacobian $\frac{\partial
  (L_1,\ldots,L_N)}{\partial (g_1,\ldots, g_N)}$ vanishes. We will now
derive a condition on the length functions that determines whether the
eigenvalue zero is related to a change in the number of positive
eigenvalues.

Let's first consider the case of two middle intervals, i.e., $N=2$ and
assume that there is some point $(\widehat g_1,\widehat g_2)$ away
from the existence bifurcation points, such that the determinant of
the Jacobian $\frac{\partial (L_1,L_2)}{\partial (g_1,g_2)}$ vanishes
at $(g_1,g_2)=(\widehat g_1,\widehat g_2)$ and hence the linearisation
about $\uf(\widehat g_1,\widehat g_2)$ has an eigenvalue zero.  We
define the values of the length functions at this point to be
$L_1(\widehat g_1,\widehat g_2)=\widehat L_1$ and $L_2(\widehat
g_1,\widehat g_2)=\widehat L_2$.  In~\cite{KDDS11} it is shown that
$\frac{\partial L_1}{\partial g_2}(g_1,g_2) = \frac{\partial
  L_2}{\partial g_1}(g_1,g_2) =\frac{1}{\mathcal{B}_1(g_1,g_2)}\neq
0$.  Hence at $(g_1,g_2)=(\widehat g_1,\widehat g_2)$, the vanishing
Jacobian implies that
\[
\frac{\partial L_1}{\partial g_1}(\widehat g_1,\widehat g_2) 
\frac{\partial L_2}{\partial g_2}(\widehat g_1,\widehat g_2)  =
   \frac{1}{\mathcal{B}_1(g_1,g_2)^2}\neq0, \mbox{ thus both } 
\frac{\partial L_1}{\partial g_1}(\widehat g_1,\widehat g_2) \neq 0,\, 
\frac{\partial L_2}{\partial g_2}(\widehat g_1,\widehat g_2)  \neq 0.
\]
Looking at the equations $L_1(g_1,g_2)=\widehat L_1$ and
$L_2(g_1,g_2)=\widehat L_2$, the inequalities above and the implicit
function theorem give that there exist smooth curves $\widetilde
g_1(g_2)$, for $g_2$ nearby $\widehat g_2$, and $\widetilde g_2(g_1)$,
for $g_1$ nearby $\widehat g_1$, such that $L_1(\widetilde
g_1(g_2),g_2)=\widehat L_1$, $\widetilde g_1(\widehat g_2) = \widehat
g_1$ and $L_2(g_1,\widetilde g_2(g_1))=\widehat L_2$, $\widetilde
g_2(\widehat g_1) = \widehat g_2$. The associate wave fronts
\[
u^1_f(g_1): = \uf(g_1,\widetilde g_2(g_1)) \mbox{ and }
u^2_f(g_2): = \uf(\widetilde g_1(g_2),g_2) 
\]
solve the wave equation with one middle interval (with length
$L=L_1+L_2$ and middle potential $V_m$ being the combination of $V_{m_1}$
and $V_{m_2}$, each at their appropriate interval) and $u^1_f(\widehat
g_1) = \uf(\widehat g_1,\widehat g_2) = u^2_f(\widehat g_2)$. So the results of
the previous section can be used to determine whether or not the
eigenvalue zero of the linearisation about the $\uf(\widehat g_1,\widehat
g_2)$ signals a change in stability. 
\begin{cor}
  Let the front $\uf(x;g_1,g_2)$ be a solution of \eqref{WaveEqN} with
  $N=2$ (and  $(g_1,g_2)$ away from the existence bifurcation points),
  such that all zeroes of $\partial_x\uf(x;g_1,g_2)$ are simple and
  the lengths of the middle intervals of $\uf(x;g_1,g_2)$ form a
  smooth length surface $L(g_1,g_2)$.  There is a change in stability
  of $\uf(x;g_1,g_2)$ at $(g_1,g_2)=(\widehat{g}_1,\widehat g_2)$ when varying
  $g_1$ or $g_2$ if and only if
\begin{enumerate}
\item the determinant of the Jacobian
  $\det\left(\frac{\partial(L_1,L_2)}{\partial(g_1,g_2)}\right)=0$,
  implying that nearby $(\widehat g_1,\widehat g_2)$ there are curves $\widetilde
  g_1(g_2)$ and $\widetilde g_2(g_1)$ with $L_1(\widetilde g_1(g_2),g_2)=\widehat
  L_1$, $\widetilde g_1(\widehat g_2) = \widehat g_1$ and $L_2(g_1,\widetilde
  g_2(g_1))=\widehat L_2$, $\widetilde g_2(\widehat g_1) = \widehat g_2$;
\item the order of the zero $\widehat{g_1}$ of $\frac{dL_1}{d
    g_1}(g_1,\widetilde g_2(g_1)))$ is odd \emph{or} the order of the zero
  $\widehat{g_2}$ of $\frac{dL_2}{d g_2}(\widetilde g_1(g_2),g_2)$ is odd;
\item the eigenfunction  $\Psi(\widehat{g}_1,\widehat g_2)$ has no zeroes.
\end{enumerate}
\end{cor}
\noindent
This result has been used implicitly in~\cite{KDDS11} to prove that
the introduction of a defect in a $0$-$\pi$ Josephson junction can
lead to the stabilisation of a \emph{non-monotonic} stationary
front (the fact that the largest eigenvalue becomes negative was
verified numerically in~\cite{KDDS11}).

\smallskip This result can be generalised to an arbitrary number of
middle intervals. In the Lemma below, it will be shown that the
vanishing of the determinant of the Jacobian $\frac{\partial
  (L_1,\ldots,L_N)}{\partial (g_1,\ldots, g_N)}$ at $(\widehat
g_1,\ldots, \widehat g_N)$ implies that there exist two curves $\widetilde
g^1_j(g_1)$ and $\widetilde g^N_j(g_N)$, $j=1,\ldots,N$, nearby
$(\widehat g_1,\ldots, \widehat g_N)$, with $L_j(\widetilde
g^i_1(g_i),\ldots,\widetilde g_N^i(g_i))=L_j(\widehat g_1,\ldots,
\widehat g_N)$, $i=1,N$ and $j\in\{1, \ldots, N\}\backslash\{i\}$.
\begin{lemma}\label{lem.curves}
  Let $\det\left(\frac{\partial (L_1,\ldots,L_N)}{\partial
      (g_1,\ldots, g_N)}(\widehat
g_1,\ldots, \widehat g_N)\right)=0$ and define $\widehat
  L_j=L_j(\widehat g_1,\ldots, \widehat g_N)$. Then there exist curves
  $\widetilde g^1_j(g_1)$, $j=1,\ldots, N$, for $g_1$ nearby $\widehat
  g_1$, and $\widetilde g^N_j(g_N)$, $j=1,\ldots, N$, for $g_N$ nearby
  $\widehat g_N$, with $\widetilde g^i_j(\widehat
  g_i)=\widehat g_j$, $j=1,\ldots, N$ and $L_j(\widetilde g^i_1(g_i),
  \ldots, \widetilde g^i_N(g_i))=\widehat L_j$, $j\in\{1, \ldots,
  N\}\backslash\{i\}$, for $i=1,N$.
\end{lemma}
\begin{proof}
  First we will show with a contradiction argument that if
  $\det\left(\frac{\partial (L_1,\ldots,L_N)}{\partial (g_1,\ldots,
      g_N)}\right)=0$, then $\det\left(\frac{\partial
      (L_2,\ldots,L_N)}{\partial (g_2,\ldots, g_N)}\right)\neq0$.
  In~\cite{KDDS11}, it is shown that the Jacobian $\frac{\partial
    (L_1,\ldots,L_N)}{\partial (g_1,\ldots, g_N)}$ is tri-diagonal
  with on the diagonal the derivatives $\frac{\partial L_i}{\partial
    g_i}$, $i=1,\ldots, N$, and on the off-diagonal the non-zero
  derivatives $\mathcal{B}_i^{-1} = \frac{\partial L_{i+1}}{\partial
    g_{i}} = \frac{\partial L_i}{\partial g_{i+1}}$, $i=1,\ldots,
  N-1$. All other derivatives $\frac{\partial L_i}{\partial g_{j}}=0$,
  $|j-i|>1$. An expansion with respect to the first row gives
\begin{equation}\label{eq.det_exp}
  \det \left(\frac{\partial(L_1,\ldots,L_N)}{\partial (g_1,\ldots,
      g_N)}\right) = 
  \frac{\partial L_1}{\partial g_1}\, 
  \det \left(\frac{\partial(L_2,\ldots,L_N)}{\partial (g_2,\ldots,
      g_N)}\right) - \frac{1}{\mathcal{B}^2_1}\, 
  \det \left(\frac{\partial(L_3,\ldots,L_N)}{\partial (g_3,\ldots,
      g_N)}\right).
\end{equation}
Thus if both $\det\left(\frac{\partial (L_1,\ldots,L_N)}{\partial
    (g_1,\ldots, g_N)}\right)=0$ and $\det\left(\frac{\partial
    (L_2,\ldots,L_N)}{\partial (g_2,\ldots, g_N)}\right)=0$, then also
$\det\left(\frac{\partial (L_3,\ldots,L_N)}{\partial (g_3,\ldots,
    g_N)}\right)=0$. Using this in the equivalent expression
to~\eqref{eq.det_exp} for $\det\left(\frac{\partial
    (L_2,\ldots,L_N)}{\partial (g_2,\ldots, g_N)}\right)$, it follows
that also $\det\left(\frac{\partial (L_4,\ldots,L_N)}{\partial
    (g_4,\ldots, g_N)}\right)=0$. We can continue this argument and
conclude that if both $\det\left(\frac{\partial
    (L_1,\ldots,L_N)}{\partial (g_1,\ldots, g_N)}\right)=0$ and
$\det\left(\frac{\partial (L_2,\ldots,L_N)}{\partial (g_2,\ldots,
    g_N)}\right)=0$, then also $\det\left(\frac{\partial
    (L_j,\ldots,L_N)}{\partial (g_j,\ldots, g_N)}\right)=0$, for
$j=3,\ldots, N$. However, as we have seen in the case $N=2$,
$\det\left(\frac{\partial (L_{N-1},L_N)}{\partial (g_{N_1},
    g_N)}\right)=0$, implies that $\frac{\partial L_N}{\partial
  g_N}\neq0$, which contradicts the previous statement for $j=N$. So
we conclude that if $\det\left(\frac{\partial
    (L_1,\ldots,L_N)}{\partial (g_1,\ldots, g_N)}\right)=0$, then
$\det\left(\frac{\partial (L_2,\ldots,L_N)}{\partial (g_2,\ldots,
    g_N)}\right)\neq0$.

By definition, $L_j(\widehat g_1, \ldots, \widehat g_N)=\widehat L_j$, $j=2,
\ldots , N$. Since $\det\left(\frac{\partial
    (L_2,\ldots,L_N)}{\partial (g_2,\ldots, g_N)}\right)\neq0$, the
implicit function theorem implies that there are curves $\widetilde
g^1_j(g_1)$, $j=2,\ldots,N$ for $g_1$ nearby $\widehat g_1$, such that
$\widetilde g^1_j(\widehat g_1)=\widehat g_j$, and $L_j(g_1, \widetilde g^1_2(g_1),\ldots,
\widetilde g^1_N(g_1))=\widehat L_j$, $j=2,\ldots,N$. If we define $\widetilde
g_1^1(g_1)=g_1$, then we have derived the statement in the Lemma for
$i=1$. In a similar way, we can prove the case for $i=N$. 
\end{proof}

Lemma~\ref{lem.curves} implies that we are again in the situation of the previous
section and we can state the following corollary about a change in
stability. 
\begin{cor}\label{cor.N} 
Let the front $\uf(x;g_1,\ldots,g_N)$ be a
  solution of \eqref{WaveEqN} (with $(g_1,\ldots,g_N)$ away from the
  existence bifurcation points), such that all zeroes of
  $\partial_x\uf(x;g_1,\ldots,g_N)$ are simple and the lengths of the
  middle intervals of $\uf(x;g_1,\ldots,g_N)$ form a smooth length
  hyper-surface $L(g_1,\ldots,g_N)$.  There is a change in stability
  for the front
  $\uf(x;g_1,\ldots,g_N)$ at $(g_1,\ldots,g_N)=(\widehat{g}_1,\ldots,\widehat
  g_N)$ when varying $g_1$ or $g_N$ if and only if
\begin{enumerate}
\item the determinant of the Jacobian
  $\det\left(\frac{\partial(L_1,\ldots,L_N)}{\partial(g_1,\ldots,g_)}(\widehat
g_1,\ldots, \widehat g_N)\right)=0$;
\item the order of the zero $\widehat{g_1}$ of $\frac{dL_1}{d
    g_1}(g_1,\widetilde g^1_2(g_1),\ldots,\widetilde g^1_N(g_1))$ is odd \emph{or}
the order of the zero $\widehat g_N$ of $\frac{dL_N}{d
    g_N}(\widetilde g^N_1(g_N),\ldots,\widetilde g^{N}_{N-1}(g_N),g_N)$ is odd;
\item the eigenfunction  $\Psi(\widehat{g}_1,\ldots,\widehat g_N)$ has no zeroes.
\end{enumerate}
\end{cor}

\section{Conclusion}
If stationary fronts of the inhomogeneous nonlinear wave
equation~\eqref{WaveEq} are considered away from any bifurcation
points, then Theorem~\ref{New1} gives a necessary and sufficient
condition for a branch of stationary fronts parametrised by the
Hamiltonian $g$ to change stability.  Furthermore,
Theorem~\ref{th.evans} identifies the stable branch and gives a
sufficient criterion for unstable branches in general.  Finally
Corollary~\ref{cor.N} gives such necessary and sufficient condition
for stability chance in a family of stationary fronts of the nonlinear
wave equation~\eqref{WaveEqN} with $N$ inhomogeneity intervals.

As illustrated in the example in section~\ref{sec.ex}, in many
specific applications these conditions can be used to determine the
stability of whole branches of solutions from the length curves.  For
instance if one knows that one branch of solutions is stable -- maybe
from an asymptotic analysis or otherwise -- then the conditions
imply a change of stability will occur when the length function has
a turning point.  If a branch was unstable, then looking at the
eigenfunction $\Psi(g)$ at turning points and checking if it has any
zeroes, will immediately show whether the adjacent branch is stable or
`more unstable' (has an extra positive eigenvalue).  

\bigskip\noindent
\textbf{Acknowledgement} 
Christopher Knight was partially supported in this work by an EPSRC
Doctoral Training Grant: EP/P50404X/1. We are also very grateful to
the referee who speculated that there could be a link between the sign
of $L'(g)$ and instability as in the VK criterion.


\begin{thebibliography}{99}

\bibitem{AGJ90} 
J. Alexander, R. Gardner, and C.K.R.T. Jones,
\newblock 
A topological invariant arising in the stability analysis of
traveling waves. 
\newblock 
J. Reine Angew.\ Math., 410, pp.\ 167-212, 
\newblock
1990.

\bibitem{BS78}
A.R.~Bishop and T.~Schneider (Eds.),
\newblock
Solitons and Condensed Matter Physics: Proceedings of a Symposium Held
June 27-29, 1978, 
\newblock
Springer-Verlag,
\newblock
1978.

\bibitem{Coppel65}
W.A.~Coppel,
\newblock 
Stability and Asymptotic Behaviour of Differential Equations.
\newblock
D.C.\ Heath \& Co: Boston,  1965.

\bibitem{Coppel78}
W.A. Coppel. 
\newblock
Dichotomies in stability theory. Lecture Notes in Mathematics, volume 629.
\newblock
Springer, 1978.

\bibitem{D85}
A.S.~Davydov,
\newblock
Solitons in Molecular Systems,
\newblock
Dordrecht, the Netherlands: Reidel,
\newblock
1985.

\bibitem{DDKS09}
G.~Derks, A.~Doelman, C.J.K.~Knight, and H.~Susanto,
\newblock
Pinned fluxons in a Josephson junction with a finite-length inhomogeneity, 
\newblock
European J.~Appl.~Math., 23(2), pp.~201-244,
\newblock
2012

\bibitem{Eastham89}
M.S.P.~Eastham,
\newblock 
The Asymptotic Solution of Linear Differential Systems. 
Clarendon Press: Oxford,  1989.


\bibitem{Evans75}
 J.W. Evans,
\newblock
Nerve axon equations. IV. The stable and unstable impulse. 
\newblock 
Indiana Univ.\ Math.\ J., 24, pp.\ 1169-1190,
\newblock
1990.

\bibitem{GJM79}
J.D.~Gibbon, I.N.~James, and I.M.~Moroz,
\newblock
The Sine-Gordon Equation as a Model for a Rapidly Rotating Baroclinic Fluid,
\newblock
Phys. Script. 20, pp.~402-408,
\newblock
1979.

\bibitem{GSS87}
M. Grillakis, J. Shatah, and W. Strauss,
\newblock  
Stability theory of solitary waves in the presence of symmetry I. 
\newblock
J. Funct.\ Anal., 74, pp.\ 160–197,
\newblock
1987

\bibitem{GSS90}
M. Grillakis, J. Shatah, and W. Strauss,
\newblock
Stability theory of solitary waves in the presence of symmetry II. 
\newblock 
J. Funct.\ Anal., 94, pp.\ 308–348, 1990.

\bibitem{KDDS11}
C.J.K.~Knight, G.~Derks, A.~Doelman, and H.~Susanto
\newblock
Stability of Stationary Fronts in a Non-linear Wave Equation with Spatial Inhomogeneity,
\newblock
Journal of Differential Equations, 254(2), pp.~408-468,
\newblock
2013.



\bibitem{PB97}
S.~Pagano and A.~Barone,
\newblock
Josephson junctions.
\newblock
Supercond. Sci. Technol. 10 pp.~904-908,
\newblock
1997.


\bibitem{PW92}
R.L. Pego and M.I. Weinstein,
\newblock
Eigenvalues, and Instabilities of Solitary Waves
\newblock
Phil.\ Trans.\ R. Soc.\ Lond.\ A, 340, pp.\ 47-94, 
\newblock
1992.


\bibitem{S69}
A.C.~Scott,
\newblock
A nonlinear Klein-Gordon equation
\newblock
Am.\ J.\ Phys.\ 37, pp.~52-61,
\newblock
1969.

\bibitem{VK73}
M.G.~Vakhitov and  A.A. Kolokolov,
\newblock
Stationary solutions of the wave equation in the medium with nonlinearity
saturation. 
\newblock 
Izv. Vyssh. Uch. Zav. Radiofiz. 16, 1020, 1973; 
\newblock 
Radiophys. Quantum Electron 16, pp.\ 783–789, 1973.


\bibitem{W74}
G.B.~Whitham,
\newblock
Linear and Nonlinear Waves,
\newblock
Wiley-Interscience,
\newblock
1974.

\bibitem{Y98}
L.V.~Yakushevich,
\newblock
Nonlinear Physics of DNA,
\newblock
Wiley series in Nonlinear Science,
\newblock
1998.
\end{thebibliography}
\end{document}